\documentclass[11]{article}
\usepackage[latin1]{inputenc}
\usepackage{amsmath}
\usepackage{amstext,amsbsy}
\usepackage{epsfig,multicol}
\usepackage{amsthm}
\usepackage{amssymb,latexsym,amsfonts}
\usepackage{graphicx}
\usepackage{caption,subcaption}
\usepackage{url}

\theoremstyle{plain}
\newtheorem{theorem}{Theorem}[section]

\newtheorem{corollary}[theorem]{Corollary}
\newtheorem{proposition}[theorem]{Proposition}

\theoremstyle{definition}

\newtheorem*{properties*}{Properties}

\newenvironment{definition*}[1][Definition]{\begin{trivlist}
\item[\hskip \labelsep {\bfseries #1}]}{\end{trivlist}}

\newcommand{\A}{\ensuremath{\textsc{a}}}
\newcommand{\C}{\ensuremath{\textsc{c}}}
\newcommand{\G}{\ensuremath{\textsc{g}}}
\newcommand{\U}{\ensuremath{\textsc{u}}}
\newcommand{\T}{\ensuremath{\textsc{t}}}

\newcommand{\poly}{\ensuremath{\mathcal{P}}}		
\newcommand{\nfan}{\ensuremath{\mathcal{N}(\mathcal{P})}}		
\newcommand{\sigs}{\ensuremath{\mathcal{S}}}		
\newcommand{\optsig}{\ensuremath{\mathcal{V}}}		
\newcommand{\R}{\ensuremath{\mathbb{R}}}		
\newcommand{\acpl}{\ensuremath{\mathcal{R}_{b_{0}}}}	
\newcommand{\sv}[1]{\ensuremath{\mathbf{#1}}}		
\DeclareMathOperator{\cone}{cone}			
\newtheorem{observation}[theorem]{Observation}

\begin{document}
\title{On the structure of RNA branching polytopes}
\author{Fidel} \author{Fidel Barrera-Cruz$^{1}$, Christine Heitsch$^1$, and Svetlana Poznanovi\'c$^2$ \\ [6pt]
$^{1}$ School of Mathematics\\
Georgia Institute of Technology \\ [6pt] 
$^{2}$ Department of Mathematical Sciences\\
Clemson University\\[5pt]
}
\date{}
\maketitle
\begin{abstract} The prevalent method for RNA secondary structure prediction for a single sequence is free energy minimization based on the nearest neighbor thermodynamic model (NNTM). One of the least well-developed parts of the model is the energy function assigned to the multibranch loops. Parametric analysis can be performed to elucidate the dependance of the prediction on the branching parameters used in the NNTM. Since the objective function is linear, this boils down to analyzing the normal fans of the \emph{branching polytopes}. Here we show that because of the way the multibranch loops are scored under the NNTM, certain branching patterns are possible for all sequences. We do this by characterizing the dominant parts of the parameter space obtained by looking at the relevant section of the normal fan;  therefore, we conclude that the structures that are normally found in nature are obtained for a relatively small set of parameters. \\
\end{abstract} 



{{\renewcommand{\thefootnote}{} \footnote{\emph{E-mail addresses}:
fidelbc@math.gatech.edu (F. Barrera-Cruz), heitsch@math.gatech.edu (C.Heitsch), spoznan@clemson.edu (S.~Poznanovi\'c)}

\footnotetext[1]{C.H. was partially supported by a BWF CASIS. S.P. was partially supported by NSF DMS-1312817.} }

%
%
\section{Introduction} \label{intro}
%
%

RNA is a chain of four nucleotides, abbreviated \A, \C, \G, and \U\ (instead of \T), which form the familiar Watson-Crick pairings, similar to DNA. Traditionally, RNA has been thought of as an important nucleic acid that plays a role in the transcription of the genetic code stored in DNA and its translation into proteins. However, in the last few decades, it has been discovered that RNA also performs other critical biological functions, including gene splicing, editing, and regulation.  Knowing the structure of noncoding RNA molecules is critical to understanding and manipulating their cellular functions  and that is why the prediction of the RNA structure has been an important problem in computational biology in the recent past.

The structure of RNA is understood hierarchically. Unlike DNA, most of RNA is found in nature as a self-complementary single strand which folds onto itself by formation of intra-sequence base pairings. The nucleotide chain itself is thought of as the \emph{primary} structure, the intra-sequence base pairs form the \emph{secondary} structure, while the \emph{tertiary} structure includes more complex interactions, including pseudoknots and base triples. Determining the tertiary structure has been challenging, both experimentally and computationally, and hence a lot of attention has been devoted to the prediction of the secondary structure. The standard approach for single-sequence secondary structure prediction is free energy minimization~\cite{mathews2006prediction}  which uses a nearest-neighbor thermodynamic model (NNTM) with several hundreds of mostly experimentally determined parameters~\cite{turner2009nndb}. However, when these minimum free energy (MFE) predictions are compared to structures derived from information-theoretic means, the current gold-standard, the average accuracy for longer ribosomal RNA sequences is only 40\%~\cite{doshi2004evaluation}. Hence, it is critical to understand which aspects of RNA base pairing are not captured well by the NNTM.

We focus here on the part of the energy function which governs the branching of an RNA secondary structure, which is known to be a weakness of the current model. For computational reasons, the entropic cost is modeled as an affine function with three parameters. A very natural question to ask is: \emph{How does the optimal secondary structure depend on the branching loop parameters?} Methods from geometric combinatorics, specifically polytopes (which we termed \emph{branching polytopes}) and their normal fans, can be used to perform a full parametric analysis of the branching part of the NNTM. The computational framework, and proof-of-principle results, which give the first complete parametric analysis of the branching part of the NNTM for real RNA sequences were presented in~\cite{drellich2017geometric}. 

The branching polytopes depend on the RNA sequence and have hundreds of vertices and facets even for sequences of fewer than 100 nucleotides~\cite{inprep}, which makes it challenging to compare them in a biologically meaningful way. However, here we prove that certain dominant features are common to all of them under some mild conditions that seem true for RNA sequences. In particular, we show that, independently of the RNA sequence, the structures that are more likely to be obtained as optima, as parameters vary, have common characteristics and we characterize these structures completely.  As a result, we demonstrate that the dominant regions of the normal fans of the branching polytopes, which are common to all RNA sequences, are a consequence of the optimization problem and that the biologically more realistic secondary structures are less likely to be obtained as optima for a large set of parameters. 

Our results build nicely on a simplified model of RNA folding defined and analyzed by Hower and Heitsch~\cite{hower2011parametric}, where some of the same type of extremes were observed. However, this is the first parametric analysis of the NNTM for real RNA sequences, in which the energies of all of the motifs are included in the same way as they are computed in the free energy minimization used to predict secondary structures.

The paper is organized as follows. In Section~\ref{preliminaries} we give the preliminaries. We give a definition of a secondary structure and we explain the part of the NNTM used to score the branching loops. Then we define the branching polytopes. In Section~\ref{results} we characterize the dominant cones of the normal fan of the branching polytopes, where we specifically focus on the trade-off between the entropic cost of forming a branching loop and the stability of a helix branching off. We end with conclusions in the last section.

%
%
\section{Preliminaries} \label{preliminaries}
%
%

\textbf{Secondary structure.} An RNA secondary structure consists of runs of stacked base pairs (helices) separated by single-stranded regions (loops).  Mathematically it is a partial non-crossing matching, which means that when the RNA sequence is written on a circle and the base pairings are drawn  as straight lines, there are no crossing lines. The \emph{exterior loop} is the loop that contains the two ends of the RNA strand. In this paper we focus on the so called \emph{multibranch loops}, i.e., the loops that have at least 3 helices meeting it. The exterior loop is not considered to be a multibranch loop, regardless of how many helices are incident with it.  

There are three types of base pairs allowed in a secondary structure: the Watson-Crick pairs \A-\U\, \C-\G\, and the wobble pair \G-\U. Even with these restrictions,  there are multiple secondary structures for a given sequence (see Figure~\ref{fig:multiple} for an example of two possible structures for a tRNA from \textit{Synechococcus sp. WH 8102\footnote[1]{\G\C\C\C\C\C\A\U\C\G\U\C\U\A\G\A\G\G\C\C\U\A\G\G\A\C\A\C\C\U\C\C\C\U\U\U\C\A\C\G\G\A\G\G\C\G\A\C\A\G\G\G\G\U\U\C\G\A\A\U\C\C\C\C\U\U\G\G\G\G\G\U\A}} generated using~\cite{zuker2003mfold}); in fact the number of secondary structures grows exponentially with sequence length~\cite{stein1979some}. 

\begin{figure}[h]
\centering
\begin{multicols}{2}
\includegraphics[width=0.6\linewidth]{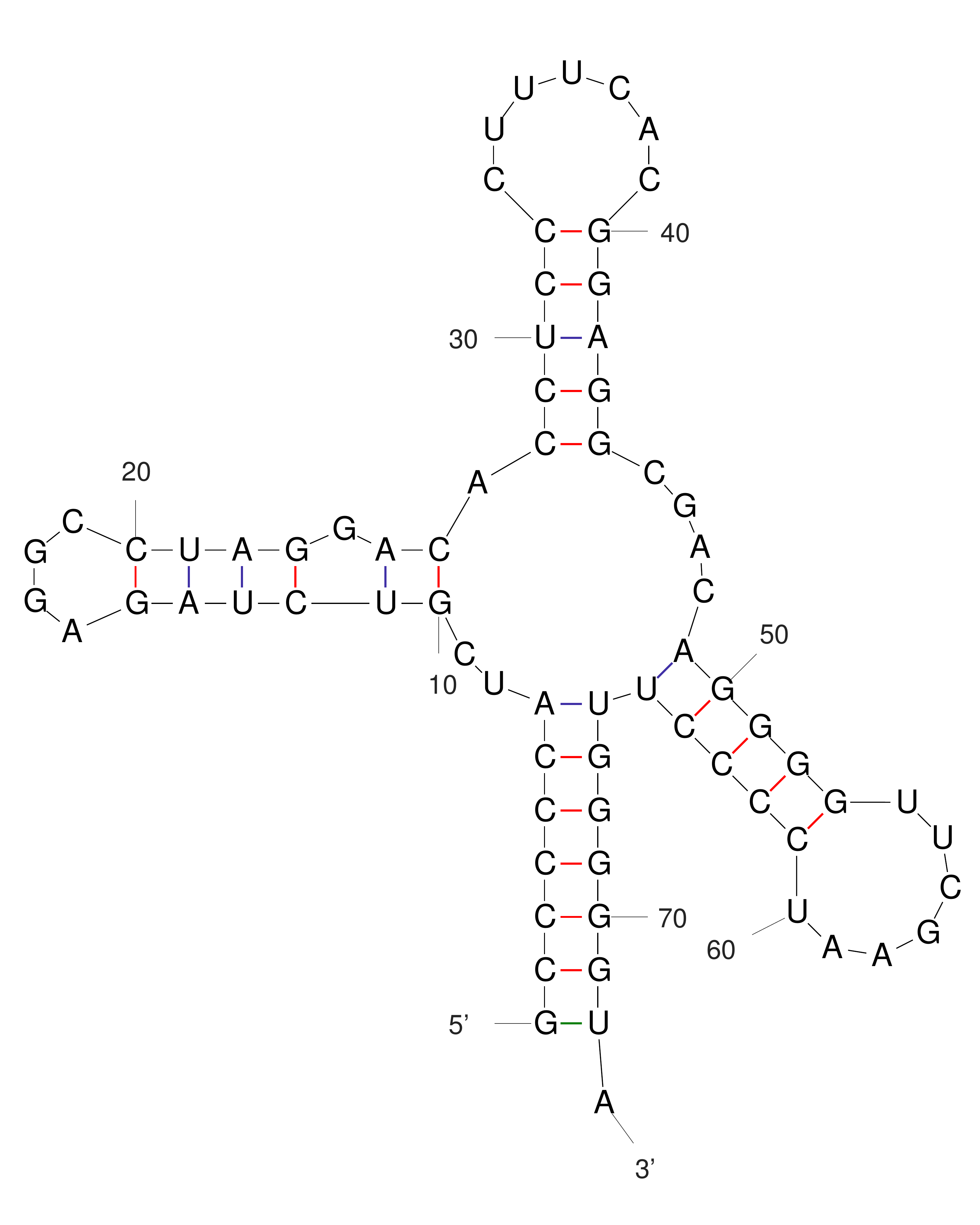}
\includegraphics[width=0.6\linewidth]{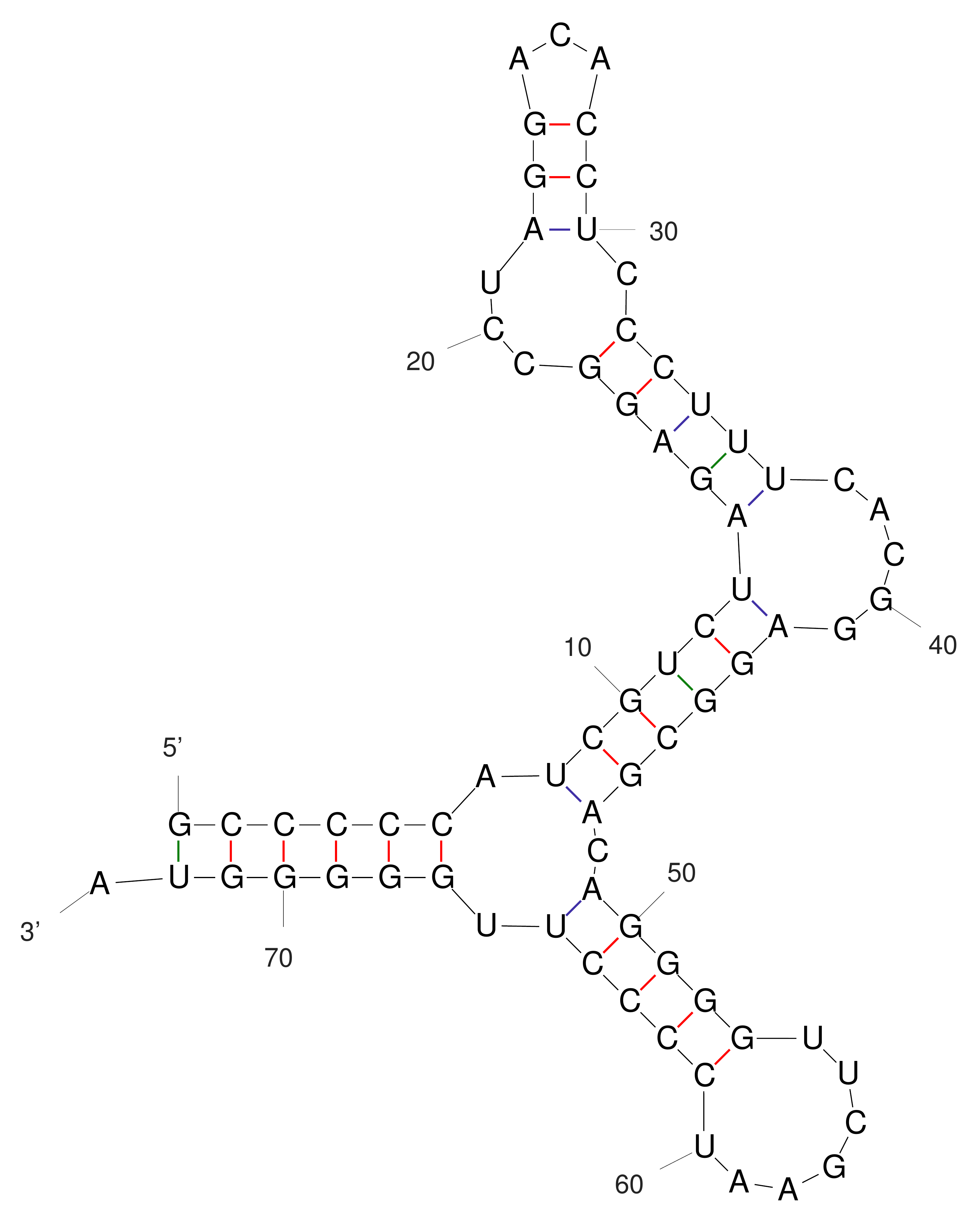}
\end{multicols}
\caption{Two of many possible structures for the same tRNA sequence.}
\label{fig:multiple}
\end{figure}

\textbf{Signatures.} The energy function used for MFE prediction has evolved substantially over the years~\cite{nussinov1978algorithms, zuker1981optimal, mathews1999expanded}, with a significant increase in the number of parameters. It is a sum of the energy of the helices and the loops; some of the contributions of the various kinds of substructures have been determined experimentally, others have been extrapolated from experimental data, and some have been learnt computationally.  In particular, in the Turner99 version of the NNTM~\cite{turner2009nndb}  the initiation term assigned to a multibranch loop is
\[ \Delta G_{\mathrm{initiation}} = a + b \cdot (\# \text{unpaired nucleotides}) + c \cdot (\# \text{branching helices}),\] where the values 
\begin{equation} \label{T99} a= 3.4, \;\; b=0 \; \; c = 0.4 \end{equation}
have been computationally determined in~\cite{mathews1999expanded}. The free energy of a secondary structure $S$ is calculated as 
\[ \Delta G_{S} = a \cdot x + b \cdot y + c \cdot z + w,\]
where $x$ is the number of multibranch loops in $S$, $y$ is the total number of single-stranded nucleotides in the multibranch loops in $S$, $z$ is the total number of helices incident with the multibranch loops in $S$ (a helix is counted twice if it is incident with two multibranch loops), and $w$ is the total sum of the energy contributions of the helices and the other kinds of loops in $S$. In light of this formula, for a given secondary structure, we define its \emph{signature} to be $(x,y,z,w)$ where $x$, $y$, $z$, $w$ are as described above. When focusing on the $x$ and $z$ coordinates, it is convenient to think of the reduced rooted plane tree representation of a secondary structure in which the nonbranching interior loops have been smoothed away. In such a representation, the root is the exterior loop, the internal nodes other than the root are the branching loops, and the leaves represent the hairpin loops. 

The signature is a map from the set of all secondary structures to $\mathbb{R}^{4}$, but even when the underlying RNA sequence is fixed, this map is not one-to-one.


\textbf{Branching polytope.} For a fixed sequence $s$ over the alphabet $\{\A, \C, \G, \U\}$, its \emph{branching polytope} \poly\ is the convex hull of all the set of signatures \sigs\ that correspond to secondary structures over $s$. We are interested in its vertices because they minimize the linear functional
\[ f_{a,b,c,d}(x,y,z,w) = a \cdot x + b \cdot y + c \cdot z + d \cdot w.\]
In particular, the  signature of the MFE structure for $s$ is a vertex of \poly\ as the optimum for $f_{a,b,c,1}$ over \sigs\ for the values of $a,b,c$ given in~\eqref{T99}. Computing $\mathrm{argmin}_{\sigs\ } f_{a,b,c,1}$ corresponds to calculating the signature of the MFE structure under an NNTM with modified multibranch parameters. While, as we said, the signature does not determine the structure completely, it contains information about its branching. The full dimensional cones of the normal fan \nfan\  of \poly\ correspond to the vertices of \poly. For a vertex $(x,y,z,w)$, we will denote by $\cone(x,y,z,w)$ the cone of parameters $(a,b,c,d)$ in \nfan\ such that $(x,y,z,w) = \mathrm{argmin}_{\sigs} f_{a,b,c,d}$. In particular, since in the NNTM we have $d=1$, we are interested in the $d=1$ slice of \nfan, which is a polyhedral subdivision of $\mathbb{R}^{3}$, and the vertices $(x,y,z,w)$ for which $\cone(x,y,z,w) \cap \{(a,b,c,d) :  d=1\}\ \neq \emptyset$.

In order to understand the trade-off between the cost $a$ of closing a multibranch loop  and the cost $c$ of starting a branching helix, we will consider the regions in 
\[ \acpl :=  \nfan\ \cap \{(a,b,c,d) \colon d=1, b=b_{0}\}.\] Figure~\ref{fig:0slices} illustrates $\mathcal{R}_{0}$ for several sequences: tRNA from \textit{Synechococcus sp. WH 8102}\footnotemark[1],  \textit{Oryza nivara}\footnote{\G\G\G\G\A\U\A\U\A\G\C\U\C\A\G\U\U\G\G\U\A\G\A\G\C\U\C\C\G\C\U\C\U\U\G\C\A\A\G\G\C\G\G\A\U\G\U\C\A\G\C\G\G\U\U\C\G\A\G\U\C\C\G\C\U\U\A\U\C\U\C\C\A}, \textit{uncultured Marinobacter sp.}\footnote{\G\G\U\C\U\G\U\A\G\C\U\C\A\G\G\U\G\G\U\U\A\G\A\G\C\G\C\A\C\C\C\C\U\G\A\U\A\A\G\G\G\U\G\A\G\G\U\C\G\G\U\G\G\U\U\C\A\A\G\U\C\C\A\C\C\C\A\G\A\C\C\C\A\C\C\A\G}, and a combinatorial sequence. The bounded regions, which are roughly around the origin, are not visible, and instead one can notice dominant unbounded regions. While the precise boundaries of the unbounded regions differ between the figures, all of them have 2 regions which dominate the first and third quadrant. These regions are separated by a sequence of \emph{unbounded stripes} in the second quadrant and a fan of \emph{unbounded wedges} in the fourth quadrant. In the next section we characterize the vertices of \poly\ which correspond to these unbounded regions.


\begin{figure}[ht]
\centering
\begin{multicols}{2}
    \includegraphics[width=0.6\linewidth]{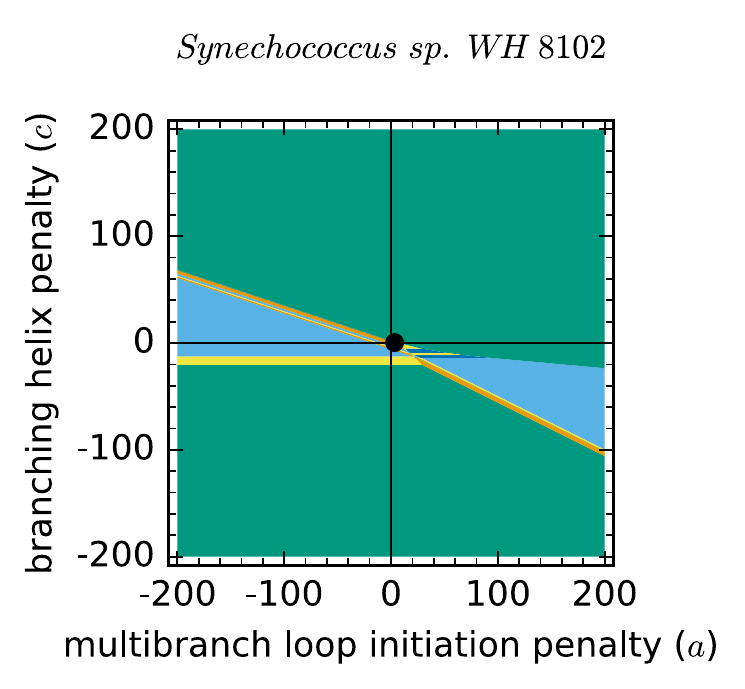} \subcaption{} 
    \includegraphics[width=0.6\linewidth]{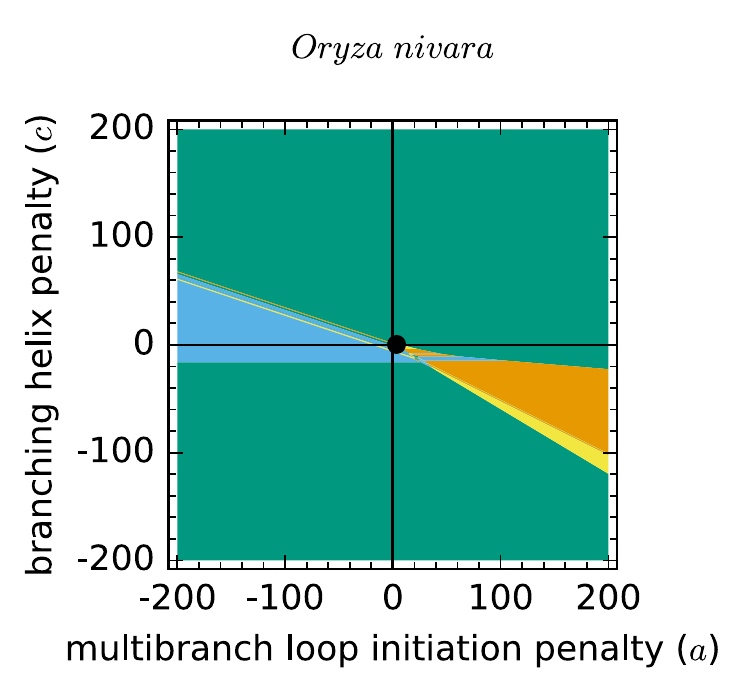} \subcaption{}
\end{multicols}
\begin{multicols}{2}
    \includegraphics[width=0.6\linewidth]{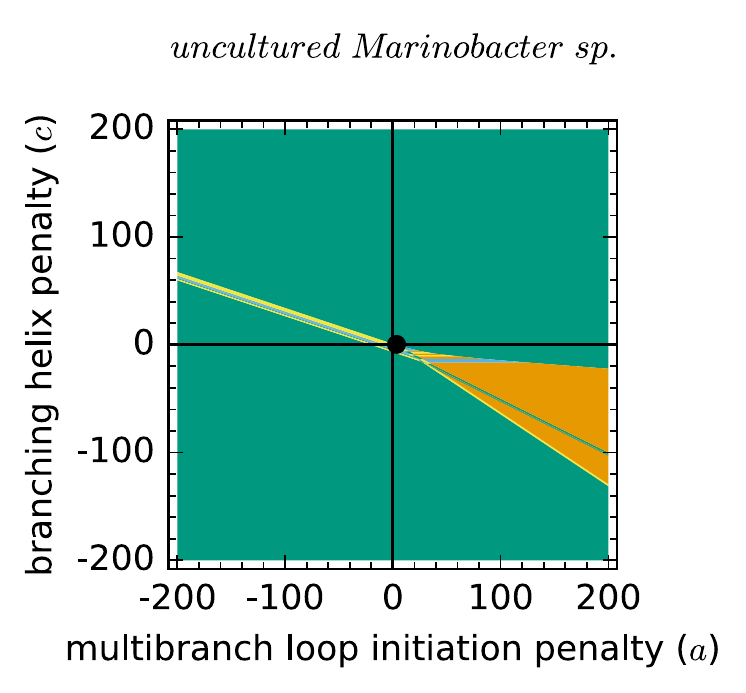} \subcaption{}
    \includegraphics[width=0.6\linewidth]{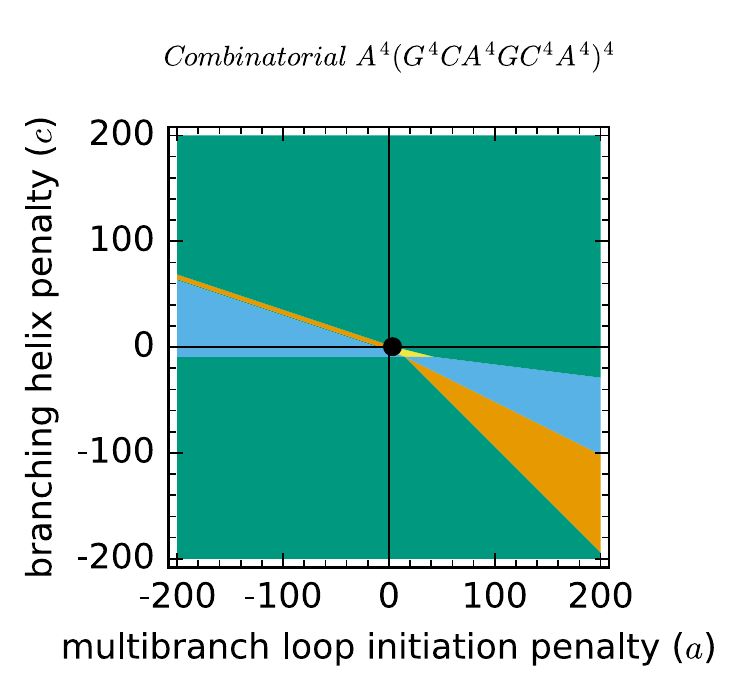} \subcaption{}
    \end{multicols}
\caption{The $\mathcal{R}_{0}$ slice of $\nfan$ for $3$ tRNA sequences and $1$ combinatorial sequence.}
\label{fig:0slices}
\end{figure}

%
%
\section{Characterization of the unbounded regions in \acpl}
\label{results}
%
%

In this section we characterize the vertices of \poly\ which correspond to the unbounded regions in \acpl. See Figure~\ref{fig: figexample} for an illustration of the outline of this section: in the $\mathcal{R}_{0}$ slice of a \textit{uncultured Thiotrichales bacterium} tRNA\footnote{\C\C\A\U\A\G\C\U\C\A\G\C\U\G\G\G\A\G\A\G\C\A\C\C\U\G\C\U\U\U\G\C\A\A\G\C\A\G\G\G\G\G\U\C\G\G\C\G\G\U\U\C\G\A\C\C\C\C\G\C\C\U\G\G\C\U\C\C\A\C\C\A\G}, each of the unbounded regions is labeled by the coordinates $(x,z)$ from the vertex $(x,y,z,w)$ which corresponds to that region.

\begin{figure}[hb]
\center
\includegraphics[width=7cm]{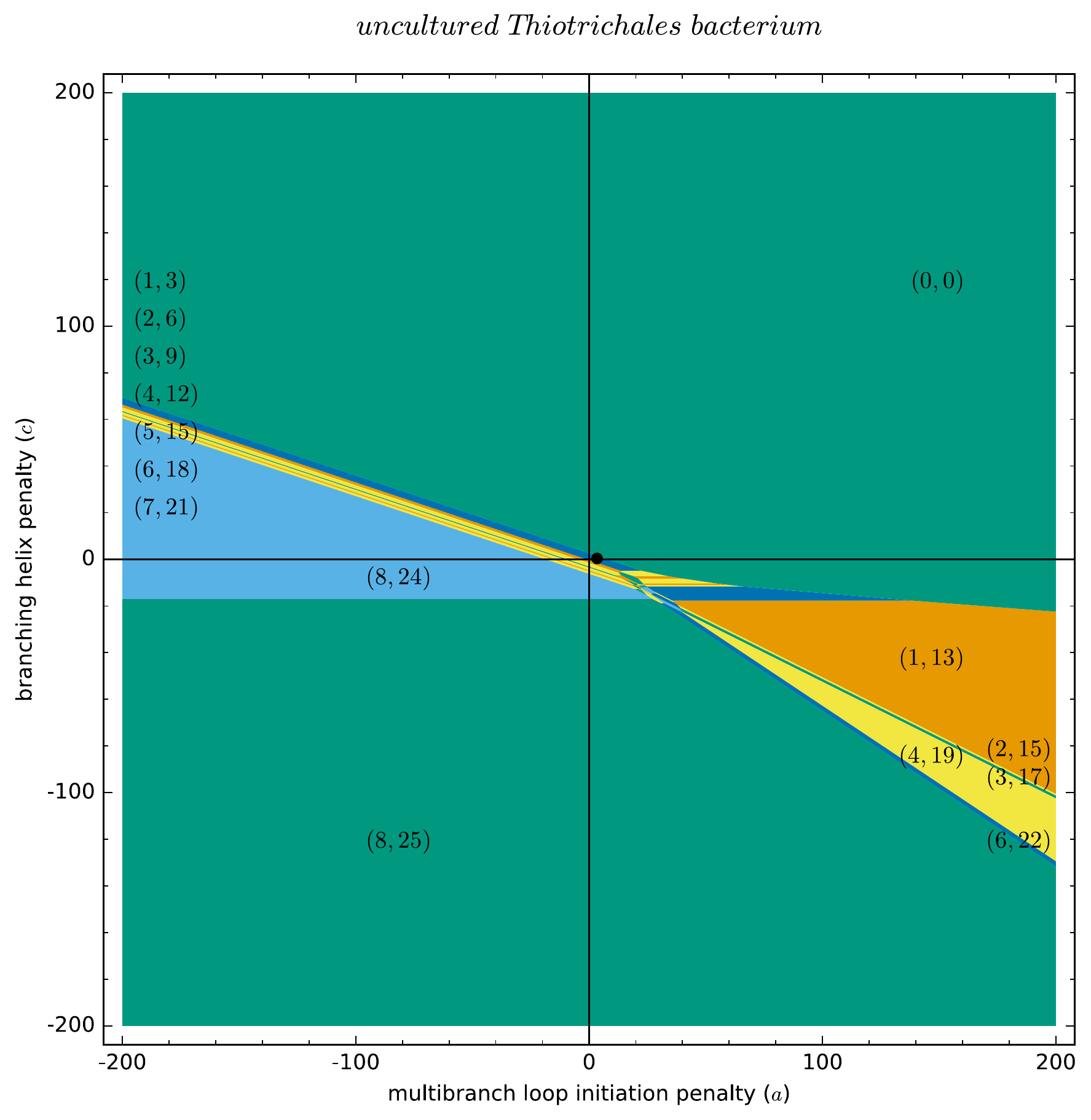}
\caption{The existence of the wedge labeled $(0,0)$ is explained by Proposition~\ref{zeroopt} and Corollary~\ref{ZeroRegion}. The wedge $(8,25)$ is explained by Corollary~\ref{maxwedge}. The wedge $(8,24)$ is explained by Proposition~\ref{xmaxzmin}. The unbounded regions $(1,3) -(7,21)$ in the NW quadrant are explained by Proposition~\ref{3xmin}. The wedge $(1,13)$ is explained by Proposition~\ref{1wedge}. Finally, the unbounded regions $(2,15) - (6,22)$ in the SE quadrant are explained by Theorem~\ref{SE}.
}
\label{fig: figexample}
\end{figure}

Let $s$ be a fixed sequence of length $n$ over the alphabet $\{\A, \C, \G, \U\}$, 
\sigs\ its set of branching signatures,
\poly\ the associated branching polytope, 
and \optsig\ the set of vertices of \poly. Let $\sv{v} = (x,y,z,w) \in \sigs$.
Let $\alpha = (a, b, c, d) \in \R^4$.
Then $\sv{v}$ is optimal for $\alpha$ if,
for all $\sv{v'} = (x',y',z',w') \in \sigs$, $\alpha \cdot (\sv{v} - \sv{v'}) \leq 0$.

Let $x_{max} = \max\{x : (x,y,z,w) \in \sigs\}$. Similarly, we define $x_{min}$, $y_{min}$, $y_{max}$, $z_{min}$, and $z_{maz}$.

\begin{proposition}
There exists $(x,y,z,w) \in \optsig$ such that $x = x_{max}$.
\label{asympmax}
\end{proposition}

\begin{proof}
This holds since $x_{max} - x' > 0$ for all 
$\sv{v'} = (x', y', z', w') \in \sigs$ with $x' < x_{max}$.
Thus, as $a \rightarrow -\infty$ for fixed $b, c, d \in \R$, 
then $\alpha \cdot (\sv{v} - \sv{v'}) \leq 0$ for 
$\sv{v} = (x_{max}, y, z, w) \in \sigs$.
Hence $\sv{v}$ is a vertex of \poly\ for some choice of $y$, $z$, $w$.
\end{proof}

There may be more than one signature in \sigs\ with $x = x_{max}$.
In this case, optimality is determined by the relationship among
the other three parameters. Clearly, comparable  results hold for $y_{max}$, $z_{max}$, and $w_{max}$.
However, since $d$ is a dummy variable in the optimization, we will henceforth consider $d = 1$ fixed, which is the only case of interest. Also, dual definitions and results hold for $x_{min}$, $y_{min}$, $z_{min}$, and $w_{min}$.
The minimum value of $0$ is achieved for $x$, $y$, and $z$
simultaneously in a structure with no branch points,
and the empty structure is one such for any sequence.
Let \[w_0 = \min\{w: (0,0,0,w) \in \sigs\}.\]

\begin{proposition}
For each $b \in \R$, there exist $a, c \in \R$ such that 
$(0,0,0, w_0)$ is optimal for $(a, b, c, 1)$.
\label{zeroopt}
\end{proposition}

\begin{proof}
Let $\sv{v_0} = (0,0,0, w_0)$ and $\alpha = (a,b,c,1) \in \R^4$.
By construction, $\alpha \cdot \sv{v_0} \leq \alpha \cdot (0,0,0,w)$ 
for every other $(0,0,0,w) \in \sigs$.
Hence we consider $\sv{v} = (x, y, z, w) \in \sigs$ with $x > 0$. 

Suppose $b \geq 0$.  
Let $a \geq 0$, $c \geq 0$ with $a + 3c \geq w_0 - w_{min}$, where
$w_{min} = \min\{w : (x,y,z,w) \in \sigs\}$ as discussed above. 
Then since the parameters are all nonnegative and 
$x \geq 1$, $y \geq 0$, $z \geq 3$, $w \geq w_{min}$,
\[ \alpha \cdot \sv{v_0} = w_0 \leq a + 3c + w_{min} \leq \alpha \cdot \sv{v}.\] 
For $b < 0$, again let $a \geq 0$, $c \geq 0$ but with 
$a + 3c \geq w_0 - w_{min} - bn$.
Then 
\[ \alpha \cdot \sv{v_0} = w_0 \leq a + b n + 3c + w_{min} \leq \alpha \cdot \sv{v}\] 
since $y \leq n$ so $b y \geq b n$.
\end{proof}

Recall that \acpl\ denotes the intersection of \nfan\
with the hyperplanes $d = 1$ and $b = b_0$.
We show that many of the characteristics visible in Figure~\ref{fig:0slices}
hold in general.
To begin, the regions of \acpl\ divide into two basic types: unbounded and not.

Let $R$ be a region in \acpl\ corresponding to $(x, y, z, w) \in \optsig$
and containing the point $(a_0, b_0, c_0, 1)$.
We already know some general bounds on $R$ as a consequence 
of Proposition~\ref{asympmax}; 
if $x < x_{max}$, then $R$ is bounded due west 
(along the ray $(a_0 - t, b_0, c_0, 1)$ for $t \geq 0$), 
with analogous conclusions for $0 < x$ and due east, 
for $0 < z$ and due north 
(along the ray $(a_0, b_0, c_0 + t, 1)$ for $t \geq 0$), 
and for $z < z_{max}$ and due south.

Those regions which are unbounded in at least one direction 
divide into two sub-types which we call ``wedges'' and ``stripes''.
Recall that a convex polyhedron is the convex sum of its vertices plus the conical sum of the direction vectors of its infinite edges. 
For an unbounded 2-dimensional polyhedron $R$, we say that $R$ is a ``stripe'' if its infinite edges have the same direction, and a ``wedge'' if its infinite edges have two different directions.

Let $\cone{(x,y,z,w)}$ denote the cone in the normal fan of \poly\ 
associated to the vertex $(x, y, z, w) \in \optsig$.
As a consequence of the proof of Proposition~\ref{zeroopt}, we have:

\begin{corollary}
For each $b_0 \in \R$, $\cone(0,0,0,w_0) \cap \acpl$ is an unbounded wedge.
 \label{ZeroRegion}
\end{corollary}

Although redundant, we retain the ``unbounded'' as an adjective to 
emphasize this as the primary characteristic, 
with the specific geometric subtype as a secondary one.

As can be seen from the choice of parameters in the proof of Proposition~\ref{zeroopt}, $\cone(0,0,0,w_0)$ dominates the NE quadrant of the $(a,c)$ plane. Moving north ($c \rightarrow \infty$) or east 
($a \rightarrow \infty$) from any point in the $(a,c)$ plane  
for $b = b_0$, $d = 1$, will eventually intersect $\cone(0,0,0, w_0)$.
Hence, there are no unbounded NE rays outside of $\cone(0, 0, 0, w_0)$. We will prove  dual statements about wedges in the SW quadrant and SW rays based on the observation that the maximum for $x$ and $z$ occur simultaneously.

For a particular $x_0 \in \R$, 
let $z_{max}(x_0)$ be the maximum number of branches for signatures 
with the given number of branch points, that is 
\[ z_{max}(x_0) = \max\{z: (x_0, y, z, w) \in \sigs\}. \]
There are obvious analogous definitions exchanging $x$ and $z$, etc.,
and dual ones for minimization.

\begin{proposition}
For each $b_0 \in \R$, there exist $y, w \in \R$ such that
$\cone(x_{max}, y, z_{max}(x_{max}), w)$ is an unbounded wedge in \acpl.
\label{xmaxzmax}
\end{proposition}

\begin{proof}
Let $x = x_{max}$, $z = z_{max}(x_{max})$ and $y, w$ be such that
$\sv{v} = (x, y, z, w) \in \sigs$ and $b_0 y + w$ is the least possible for
the given $b_0$.
Let $\alpha = (a_{0}, b_{0}, c_{0},1) \in \acpl$ and
$\sv{v'} = (x',y',z',w') \in \sigs$.
We may assume that either $x' = x$ and $3 x \leq z' < z$ or that
$x > x' \geq 0$.
We will use that $w \geq w_{min}$,  $n \geq y \geq 0$,
and  $z \leq z_{max}$.

Suppose $b_0 \geq 0$.
Then $w_{min} - w - b_0 y \leq 0$.
Let $a_{0}, c_{0}$ be such that
 \[a_{0} \leq 0, \; c_{0} \leq w_{min} - w - b_0 y, \; 
a_{0} + c_{0}(z - z_{max}) \leq w_{min} - w - b_{0} y.\]

\item If $x'=x$ and $z' \leq z-1$ then using the upper bound for $c_{0}$,
\begin{align*}(a_{0}, b_{0}, c_{0},1) \cdot (x', y', z', w') &\geq (a_{0}, b_{0}, c_{0},1) \cdot (x, 0, z-1, w_{min}) \\ &\geq   (a_{0}, b_{0}, c_{0},1) \cdot (x, y, z, w).\end{align*}
\item If $x' \leq x -1$ then from the choice of $a_0$ and $c_0$ and the fact that $c_{0} \leq 0$ we have
\begin{align*}(a_{0}, b_{0}, c_{0},1) \cdot (x', y', z', w') &\geq (a_{0}, b_{0}, c_{0},1) \cdot (x-1, 0, z_{max}, w_{min}) \\ & \geq  (a_{0}, b_{0}, c_{0},1) \cdot (x, y, z, w).\end{align*}

So, $\alpha \cdot (\sv{v} - \sv{v'}) \leq 0$. Now suppose $b_{0} <0$.
Then $w_{min} - w + b_0 (n - y) \leq 0$.
Let $a_{0}, c_{0}$ be such that
 \[a_{0} \leq 0, \; c_{0} \leq  w_{min} - w  + b_{0}(n - y), \; 
a_{0} +c_{0} (z - z_{max}) \leq w_{min} - w + b_{0}(n-y).\]
Now  $\alpha \cdot (\sv{v} - \sv{v'}) \leq 0$
follows from the choice of $c_0$ by comparison with
$\alpha \cdot (x, n, z - 1, w_{min})$ if $x' = x$ and $z' \leq z - 1$
and from the choice of $a_0$ and $c_0$ by comparison with
$\alpha \cdot (x-1, n, z_{max}, w_{min})$ if $x > x' \geq 0$.
\end{proof}

Using the same argument in which the roles of $x$ and $z$ are swapped, we can conclude that a part of the SW quadrant belongs to another unbounded wedge.

\begin{proposition}
For each $b_0 \in \R$, there exist $y, w \in \R$ such that
$\cone(x_{max}(z_{max}), y, z_{max}, w)$ is an unbounded wedge in \acpl.
\label{zmaxxmax}
\end{proposition}

The wedges from Propositons~\ref{xmaxzmax} and~\ref{zmaxxmax} can coincide, which is indeed the case in all of the examples we have seen~\cite{inprep}. Namely, all of the sequences have the following property.


\begin{observation}
We have $z_{max} = z_{max}(x_{max})$ or, equivalently, $x_{max} = x_{max}(z_{max})$.
\label{observmax}
\end{observation}

In contrast to the situation with $(0,0,0, w_0)$, however, even under the assumption from Observation~\ref{observmax}, there may 
be different optimal signatures $(x_{max}, y, z_{max}, w) \in \optsig$ depending on the choice of $b$. To summarize, we have the following dual of Corollary~\ref{ZeroRegion}.

\begin{corollary}
If Observation~\ref{observmax} holds, then for each each $b_0 \in \R$
there exist $y_{m}, w_{m} \in \R$
such that $\cone(x_{max}, y_{m}, z_{max}, w_{m}) \cap \acpl$ is an unbounded wedge. 
\label{maxwedge}
\end{corollary}

While the existence of such a wedge in each $\acpl$  follows from Propositions~\ref{xmaxzmax} and~\ref{zmaxxmax}, the proofs  additionally describe its geometry. Namely, for $b_{0} \geq 0$, the nonempty $\cone(x_{max}, y_m, z_{max}, w_m) \cap \acpl$ contains the region 
\[a_{0}, c_{0} \leq w_{min} - w_{m} - b_0 y_{m}\] and for $b_{0} < 0$ it contains the region\[a_{0}, c_{0} \leq w_{min} - w_{m}  + b_{0}(n - y_{m}).\]  Thus, if it exists,  $\cone(x_{max}, y_m, z_{max}, w_m)$ dominates the SW quadrant of $\acpl$.

%
%

We have already shown that among the edges of the unbounded regions of \acpl\ there are no unbounded NE rays. Observation~\ref{observmax} implies the dual statement about SW rays.  As a consequence, the infinite edges of the unbounded regions of \acpl\ conform to a particular geometry.

\begin{theorem}
If, and only if, Observation~\ref{observmax} holds, then an infinite edge 
for an unbounded region of $\acpl$ cannot have positive slope.
\label{negslope}
\end{theorem}

\begin{proof} 
The  proof of Proposition~\ref{zeroopt}  implies that the wedge $\cone(0,0,0, w_0) \cap \acpl$ contains the angle from $0$ to $\pi/2$ for any point in the region. This means that no unbounded region can have a NE edge.
 
 If Observation~\ref{observmax} holds, by Corollary~\ref{maxwedge}, there exist $y, w \in \R$
such that $\cone(x_{max}, y, z_{max}, w) \cap \acpl$ is an unbounded wedge. Moreover,  as we observed above, the proofs of Propositions~\ref{xmaxzmax} and~\ref{zmaxxmax}
indicate that this wedge contains the angle from $\pi$ to $3 \pi/2$ for any point in the region. This means that no unbounded region can have a SW edge. On, the other hand, if $z_{max} \neq z_{max}(x_{max})$, then the SW quadrant of  $\acpl$ contains unbounded portions of at least two wedges (possibly more): namely $\cone(x_{max}, y', z_{max}(x_{max}), w')$ and $\cone(x_{max}(z_{max}), y'', z_{max}, w'')$, for some $y', w', y'', w'' \in \R$, and hence these wedges have an unbounded SW edge.
\end{proof}

In general, the slopes of finite boundary edges are also negative.  
However, horizontal and vertical edges are seen and, though more rare,
positive slopes for a bounded region of \acpl\ have been 
observed. See~\cite{inprep} for numerical results about the bounded regions.

We next describe the unbounded regions that we see as we traverse  \acpl\ counter-clockwise from $(0,0,0, w_0)$
around to $(x_{max}, y_m, z_{max}, w_m)$. We start by proving that crossing a region boundary in the $(a,c)$ plane in either 
a horizontal or vertical direction implies a strict change in the
number of branching points or of branches, respectively, for the 
associated signatures.

\begin{proposition}
Suppose $(x, y, z, w), (x', y', z', w') \in \optsig$ are optimal for 
$(a, b_0, c, 1)$ and $(a', b_0, c', 1)$ respectively, where these
parameters lie in the interior of two distinct regions of \acpl.
If $a = a'$ but $c < c'$, then $z > z'$.
Similarly, if $c = c'$ but $a < a'$, then $x > x'$.
\label{boundary}
\end{proposition}

\begin{proof}
By assumption, $\alpha \cdot (\sv{v} - \sv{v'}) < 0$ and 
$\alpha' \cdot (\sv{v'} - \sv{v}) < 0$ for $\sv{v} = (x, y, z, w)$, 
$\sv{v'} = (x', y', z', w')$, $\alpha = (a, b_0, c, 1)$, and 
$\alpha' = (a', b_0, c', 1)$.
Hence, $(a - a')(x - x') + (c - c')(z - z') < 0$.
\end{proof}

In comparison to Proposition~\ref{asympmax}, 
Proposition~\ref{boundary} says that the number of branch points
in the corresponding optimal signatures 
increases monotonically whenever a region boundary is crossed 
as $a \rightarrow -\infty$ for fixed $c$ in \acpl.
Analogous statements can be made for the total number of branches,
and for minimization.

The proof also shows that a particular combination of $x$ and $z$ can be 
associated with at most one region of \acpl with a nonempty interior.

The consequence of Thereom~\ref{negslope} is that as we traverse 
the unbounded regions of \acpl\ counter-clockwise from $(0,0,0, w_0)$
around to $(x_{max}, y_m, z_{max}, w_m)$, we see a correlated increase 
in both $x$ and $z$. Similarly, with a clockwise traversal. The difference is that counter-clockwise, the number of branches is minimized
whereas clockwise it is maximized (subject to some other conditions).

The cones which correspond to $x_{max}$ and intersect \acpl\ all produce unbounded regions in \acpl. In particular, we have the following wedge.
\begin{proposition}
For each $b_0 \in \R$, there exist $y, w \in \R$ such that 
$\cone(x_{max}, y, z_{min}(x_{max}), w)$ is an unbounded wedge in \acpl.
\label{xmaxzmin}
\end{proposition}

\begin{proof}

Let $x = x_{max}$, $z = z_{min}(x_{max})$ and $y, w$ be such that 
$\sv{v} = (x, y, z, w) \in \sigs$ and $b_0 y + w$ is the least possible for 
the given $b_0$. Let $\alpha = (a_{0}, b_{0}, c_{0},1) \in \acpl$ and
$\sv{v'} = (x',y',z',w') \in \sigs$. Recall that $n$ is the length of sequence $s$ and $w_{min}$ is minimal
over all signatures for $s$. Hence $w \geq w_{min}$ and $n \geq y \geq 0$.

We claim $(x,y,z,w)$ is optimal for parameters 
$\alpha = (a_{0}, b_{0}, c_{0},1) \in \acpl$ where $a_0$ and $c_0$ satisfy
the constraints below. By the choice of $y$ and $w$, we may assume that for $\sv{v'} = (x',y',z',w') \in \sigs$ either $x' = x$ and $z' \geq z + 1$ or $x > x' \geq 0$.

If $b_0 \geq 0$, $w_{min} - w - b_0 n \leq 0$. Let $a_{0}, c_{0}$ be such that
 \[a_{0} \leq 0, \; c_{0} \geq   b_{0}n+w - w_{min}, \; 
a_{0} +c_{0}z \leq w_{min} - w - b_{0} n.\]
If $x' = x$ and $z' \geq z + 1$, $\alpha \cdot (\sv{v} - \sv{v'}) \leq 0$ because of the upper bound for $c_{0}$ by comparison with $(x, 0, z+1, w_{min})$. If $x > x' \geq 0$, then $\alpha \cdot (\sv{v} - \sv{v'}) \leq 0$ by comparison with $(x-1, 0, 0, w_{min})$ due to the choice of $a_0$ and $c_0$ and the fact that $c_{0} \leq 0$.

Now suppose $b_{0} <0$. 
Then $w_{min} - w + b_0 (n-y) \leq 0$.
Let $a_{0}, c_{0}$ be such that
 \[a_{0} \leq 0, \; c_{0} \geq   b_{0}(y-n)+w - w_{min}, \; 
a_{0} +c_{0}z \leq w_{min} - w + b_{0}(n-y).\] 
In this case $\alpha \cdot (\sv{v} - \sv{v'}) \leq 0$
follows from the choice of $c_0$ by comparison with 
$\alpha \cdot (x, n, z + 1, w_{min})$ if $x = x'$ and $z' \geq z + 1$
and from the choice of $a_0$ and $c_0$ by comparison with 
$\alpha \cdot (x-1, n, 0, w_{min})$ if $x > x' \geq 0$. 
\end{proof}

By the choice of parameters, we see that 
$\cone(x_{max}, y, z_{min}(x_{max}), w)$ always intersects the NW 
quadrant of \acpl.
From the choosen $a_0$ and $c_0$, we know that it is unbounded to the 
west along the line $(a_0 - t, b_0, c_0, 1)$ and, 
since $z = z_{min}(x_{max})> 0$ for $x_{max} > 0$, 
to the NW  along the line $(a_0 - t, b_0, c_0 + t/z, 1)$ for $t \geq 0$. 

Moreover, one can readily see that if $z_{min}(x_{max}) \neq z_{max}(x_{max})$ and  $\cone(x_{max}, y, z, w)$ contains $(a_{0}, b_{0}, c_{0}, 1)$ for $z_{min}(x_{max}) < z < z_{max}(x_{max})$  then it contains the whole ray $(a_0 - t, b_0, c_0, 1)$, for $t \geq 0$. Therefore $\cone(x_{max}, y, z, w)$ is an unbounded stripe between the wedges $\cone(x_{max}, y, z_{min}(x_{max}), w)$ and $\cone(x_{max}, y, z_{max}(x_{max}), w)$.

We now show that, for a given number of branch points, 
having the minimal number of branches possible is necessary for
signatures which correspond to regions that are unbounded to the NW. Dually, for a given number of branch points, having the maximal number of branch points possible is necessary for the corresponding region to be unbounded to the NW.

\begin{proposition}

Let $(x,y,z,w) \in \optsig$ such that  $R = \cone(x,y,z,w) \cap \acpl \neq \emptyset$.
If $x < x_{max}(z)$ or $z > z_{min}(x)$, then $R$ is bounded to the NW.

\label{northwest}
\end{proposition}

\begin{proof}
Suppose $\sv{v} = (x, y, z, w)$ is optimal for $\alpha = (a_0, b_0, c_0, 1)$
where  $z > z_{min}(x)$.
By definition,
there exist $y', w'$ such that $\sv{v'} = (x, y', z', w') \in \sigs$
for $z' = z_{min}(x)$. 

Let $m > 0$ and consider $\alpha' = (a_0 - t, b_0, c_0 + m t, 1)$ 
for $t > 0$.
Then 
\[ \alpha'(\sv{v} - \sv{v'}) = b_0 (y - y') + (c_0 + m t)(z - z') +
(w - w') = \alpha(\sv{v} - \sv{v'}) + m t (z - z') > 0 \] 
for $mt > 0$ sufficiently large since $\alpha(\sv{v} - \sv{v'})$ is fixed
and $z - z' > 0$.
Hence, $R$ cannot contain the 
ray $l = (a_0 - t, b_0, c_0 + m t, 1)$ for $t \geq 0$. We get the same contradiction if $x < x_{max}(z)$ and we consider a point $\sv{v''} = (x'', y'', z, w'') \in \sigs$ with $x'' = x_{max}(z)$.
\end{proof}

We know that $z_{min}(x) \geq 3 x$ since a branch point must
have at least three branches, by definition.
Hence, a minimally branched structure resembles a binary tree in the sense 
that each branch point has exactly two children.
We note, though, that this says nothing about nonbranching vertices 
and also that the root does not count as a branch point for 
our purposes, since its energy function has no entropic penalty.

As far as we have seen, this lower bound on the total number of branches
is always achieved by some signature having the maximum number of branch 
points, and it again has interesting geometric implications.

\begin{observation}
We have $z_{min}(x_{max}) = 3 x_{max}$.
\label{observmin}
\end{observation}

\begin{corollary} If Observation~\ref{observmin} holds then, for every $0 < x < x_{max}$, $z_{min}(x) = 3x$.
\end{corollary}

\begin{proof}
Suppose $S$  is a structure with signature $(x,y,3x,w)$ for some $0 < x \leq x_{max}$, $y,w \in \R$, then in the rooted tree representation of $S$, all its branching points have exactly two children. Take one of the branching points whose children are leaves (i.e. not branching points themselves) and unpair all the basepairs in the branches that meet at that node. The resulting structure has $x-1$  branching points and $3x -3$ branches.  
\end{proof}

In this case, corresponding to a signature having the minimal number of branches possible is also sufficient for
the region to be unbounded to the NW.

\begin{proposition}  
Suppose $(x, y, 3x, w) \in \optsig$ with 
$R = \cone(x,y,3x,w) \cap \acpl \neq \emptyset$.
Then $R$ is unbounded to the northwest.
\label{3xmin}
\end{proposition}

\begin{proof}
Suppose $\sv{v} = (x,y,3x,w)$ is optimal for 
$\alpha = (a_{0}, b_{0}, c_{0}, 1)$.
Let $\alpha' = (a_{0}-t, b_{0}, c_{0}+\frac{t}{3},1)$ for $t > 0$
and $\sv{v'} = (x',y',z',w') \in \optsig$. 
Since $\alpha(\sv{v} - \sv{v'}) \leq 0$, 
then $\alpha'(\sv{v} - \sv{v'}) \leq 0$ follows from
\[-tx'+\frac{t}{3}z' \geq -tx + \frac{t}{3}3x
 \]
which is equivalent to \(z' \geq 3x'\).
\end{proof}

We begin our characterization of the east half of \acpl\  with a result analogous to Proposition~\ref{xmaxzmin}; the proof is a straight-forward dualization..
\begin{proposition}
For each $b_0 \in \R$, there exist $y, w \in \R$ such that 
$\cone(x_{min}(z_{max}), y, z_{max}, w)$ is an unbounded wedge in \acpl.
\label{zmaxxmin}
\end{proposition}

By duality, $\cone(x_{min}(z_{max}), y, z_{max}, w)$ is partly in the SE quadrant. In the examples of RNA sequences we have seen~\cite{inprep}, $x_{min}(z_{max}) = x_{max}$, and  this wedge coincides with $\cone(x_{max}, y_{m}, z_{max}, w_{m})$. However,  this is not true for all sequences. For example, for the sequence $(\G\A\C\A\A\A)^6$, $z_{max}=6$, $x_{max}=2$, but $x_{min}(6)=1$. Regardless, if the sequence is long enough so that $x_{max}>1$, the SE quadrant is guaranteed to contain at least one more unbounded wedge that we haven't mentioned so far, which we show next.  

\begin{proposition}
If $x_{max} \geq 1$, then 
for each $b_0 \in \R$, there exist $y, w \in \R$ such that 
$\cone(1, y, z_{max}(1), w)$ is an unbounded wedge in \acpl.
\label{1wedge}
\end{proposition}

\begin{proof}
Let $z_1 = z_{max}(1)$ 
and $y, w$ be such that $\sv{v} = (1, y, z_1, w) \in \sigs$ and 
$b_0 y + w$ is the least possible.
Let $\alpha = (a_0, b_0, c_0, 1) \in \acpl$ and 
$\sv{v'} = (x',y',z', w') \in \sigs$.
By choice of $b_0 y + w$,
we may assume that either $x' = 1$ and $z' \leq z_1 - 1$ or that $x' \geq 2$.

Suppose $b_0 \geq 0$.  
Then, as in previous proofs, $w_{min} - w - b_0 y \leq 0$.
Let $a_{0}, c_{0}$ be such that
 \[a_{0} \geq 0, \; c_{0} \leq w_{min} - w - b_0 y, \;
a_{0} + c_{0}(z_{max} - z_1) \geq b_0 y + w - w_{min}.\]
In this case $\alpha \cdot (\sv{v} - \sv{v'}) \leq 0$
follows from the choice of $c_0$ by comparison with
$\alpha \cdot (1, 0, z_1 - 1, w_{min})$ if $x' = 1$ and $z' \leq z_1 - 1$
and from the choice of $a_0$ and $c_0$ by comparison with
$\alpha \cdot (2, 0, z_{max}, w_{min})$ if $x' \geq 2$.

Now suppose $b_{0} < 0$.
Again we have $w_{min} - w + b_0 (n - y) \leq 0$.
Let $a_{0}, c_{0}$ be such that
 \[a_{0} \geq 0, \; c_{0} \leq  w_{min} - w  + b_{0}(n - y), \;
a_{0} +c_{0} (z_{max} - z_1) \geq b_0(y-n) + w - w_{min}.\]
Now  $\alpha \cdot (\sv{v} - \sv{v'}) \leq 0$
follows from the choice of $c_0$ by comparison with
$\alpha \cdot (1, n, z_1 - 1, w_{min})$ if $x' = 1$ and $z' \leq z_1 - 1$
and from the choice of $a_0$ and $c_0$ by comparison with
$\alpha \cdot (2, n, z_{max}, w_{min})$ if $x' \geq 2$.
\end{proof}

Result analogous to Proposition~\ref{northwest} also holds; the proof is a straight-forward dualization.


%

\begin{proposition}
Let $(x,y,z,w) \in \optsig$ such that  $R = \cone(x,y,z,w) \cap \acpl \neq \emptyset$.
If $x >x_{min}(z)$ or $z < z_{max}(x)$, then $R$ is bounded to the SE. 
\end{proposition}

Hence only the regions corresponding to signatures in the set 
\[\sigs_{0}:= \{(x,y,z,w) \in \sigs \colon x=x_{min}(z), z=z_{max}(x)\}\]
are candidates for regions in \acpl\ which are 
unbounded to the southeast. Moreover, using the same reasoning, $(x,y,z,w) \in \sigs_{0}$ is bounded unless
\begin{align}
z_{max}(x') < z \;\; &\text{ for all } \;\; 1 \leq x' < x, \label{property1}\\
x_{min}(z') > x \;\; &\text{ for all }  \;\; z < z' \leq z_{max} \label{property2}.
\end{align}
In fact,~\eqref{property1} and~\eqref{property2} are equivalent. To exclude the signatures from $\sigs_{0}$ that do not satisfy~\eqref{property1} and~\eqref{property2}, it is sufficient to check against points from $\sigs_{0}$. Namely, $(x,y,z,w) \in \sigs_{0}$ satisfies~\eqref{property1} and~\eqref{property2} if and only if
\begin{equation} \label{property3}
\text{for every } \; \; (x',y',z',w') \in \sigs_{0},  \;\; x' < x \iff z' < z.
\end{equation}
Notice that for two points $(x,y,z,w), (x',y',z',w') \in \sigs_{0}$, we have $x=x'$ if and only if $z=z'$. Hence, a signature $(x,y,z,w) \in \sigs_{0}$ satisfies~\eqref{property3} if and only if
\begin{equation} \label{property4}
\text{for every } \; \; (x',y',z',w') \in \sigs_{0},  \;\; x' > x \iff z' > z.
\end{equation}
It is clear that~\eqref{property1} and~\eqref{property2} imply~\eqref{property3}. To see the converse, suppose $(x,y,z,w) \in \sigs_{0}$ satisfies~\eqref{property3} and suppose $x_{1}$ is such that $1 \leq x_{1} < x$ but $z_{1}= z_{max}(x'_{1}) \geq z$. Let 
\[ x_{2} = x_{min}(z_{1}), z_{2} = z_{max}(x_{2}), x_{3}=x_{min}(z_{2}), z_{3}=x_{min}(x_{3}), \dots.\] Then
\[x > x_{1} \geq x_{2} \geq x_{3} \geq \cdots,\]
\[z \leq z_{1} \leq z_{2} \leq z_{3} \leq \cdots\]
are two bounded sequences and, therefore, eventually stabilize. Suppose for all sufficiently large $n$, $x_{n} = x^{*}$ and $z_{n} = z^{*}$. Then there is a signature $(x^{*}, y^{*}, z^{*},w^{*}) \in \sigs_{0}$ such that $x^{*} < x$ but $z^{*} \geq z$. This is a contradiction. Therefore, we conclude that~\eqref{property3} implies~\eqref{property1}. 
Hence only the regions corresponding to signatures in the set 
\[\sigs_{1}:= \{(x,y,z,w) \in \sigs_{0} \colon (x,y,z,w) \text{ satisfies } \eqref{property3}\}\]
are candidates for regions in \acpl\ which are  unbounded to the southeast. The next result completely characterizes which ones among these are unbounded to the SE. 

\begin{theorem} \label{SE}
Suppose $(x, y, z, w) \in \sigs_{1}$, $x>1$, $z < z_{max}$ is such that $R = \cone(x, y, z, w) \cap \acpl \neq \emptyset$.  Then $R$ is bounded to the southeast if and only if there exist $(x', y', z',w'), (x'',y'',z'',w'') \in\sigs_{0}$ with  $x' < x < x'' $ (equivalently, $z' < z < z'' $) such that 
\begin{equation} \frac{x - x'}{z - z'} > 
\frac{x - x''}{z - z''}.
\label{bc}
\end{equation}
\end{theorem}

\begin{proof}
Suppose $R$ is unbounded to the southeast and contains the ray $(t, 0, -mt, 0), t\geq 0$ for some $m > 0$. Let $\sv{v_{0}}=(x_{0},y_{0},z_{0},w_{0})$. Then  $(1, 0,-m,0) \cdot (\sv{v} -\sv{v_{0}}) \leq 0$ for every $\sv{v_{0}} \in \sigs$, which implies that $\frac{x-x'}{z-z'} \leq m$ for every $(x',y',z',w') \in \sigs$ with $z' < z$ and $ m \leq \frac{x-x''}{z-z''}$ for every $(x'',y'',z'',w'') \in \sigs$ with $z'' > z$. 

Converesely, suppose there are no two points $(x', y', z',w'), (x'',y'',z'',w'') \in\sigs_{0}$ with  $x' < x < x'' $ (equivalently, $z' < z < z'' $) such that~\eqref{bc} holds. Let $m > 0$ be such that
\begin{equation} \max\left\{ \frac{x-x'}{z-z'}: (x',y',z',w') \in \sigs_{0}, z' < z\right\} \leq m \leq \min\left\{ \frac{x-x''}{z-z''}: (x'',y'',z'',w'') \in \sigs_{0}, z'' > z\right\}.
\label{bndc}
\end{equation}
The parameter $m $ can be chosen to be positive because all the fractions in the sets in~\eqref{bndc} are positive. We claim that the region $R$ contains the ray $(t, 0, -mt, 0)$ and hence is unbounded to the southeast. Suppose this is not the case. Then there is $(x_{0}, y_{0}, z_{0},w_{0}) \in \sigs$ such that 
\[ (x-x_{0}) - m (z- z_{0}) > 0.\]
We first consider the case $z_{0}<z$. Then $m >0$ implies $x > x_{0}$. Let \[ x_{1} = x_{min}(z_{0}), z_{1} = z_{max}(x_{1}), x_{2}=x_{min}(z_{1}), z_{2}=z_{max}(x_{2}), \dots.\] This way we get two monotone sequences:
\[x > x_{0} \geq x_{1} \geq x_{2} \geq \cdots \]
\[ z_{0} \leq z_{1}  \leq z_{2} \leq \cdots. \] 
which, since they are bounded, eventually stabilize. Suppose $x_{n}=x^{*}$, $z_{n}=z^{*}$ for all $n \geq n_{0}$ for some $n_{0} \in \mathbb{N}$. Then there is a signature $(x^{*}, y^{*},z^{*},w^{*}) \in \sigs_{0}$ for some $y^{*},w^{*}$. Since $(x, y, z, w) \in \sigs_{1}$, $x^{*} <x$ implies $z^{*} < z$ and, consequently, $z_{n} \leq z$ for all $n \in \mathbb{N}$. Moreover, by construction,
\[ 0 < m < \frac{x-x_{0}}{z-z_{0}} \leq \frac{x-x_{1}}{z-z_{1}} \leq \cdots \leq \frac{x-x^{*}}{z-z^{*}},\] which contradicts~\eqref{bndc}. The case when $z_{0}>z$   leads to a similar contradiction, while $z_{0}=z$ is clearly impossible.
\end{proof}

The proof of Theorem~\ref{SE} also gives a criterion for determining whether for $(x, y, z, w) \in \sigs_{1}$ the unbounded region $R = \cone(x, y, z, w) \cap \acpl$ is an unbounded stripe. Namely $R$ is a stripe if and only if 
\[\max\left\{ \frac{x-x'}{z-z'}: (x',y',z',w') \in \sigs_{0}, z' < z\right\} = \min\left\{ \frac{x-x''}{z-z''}: (x'',y'',z'',w'') \in \sigs_{0}, z'' > z\right\}.\]

Another property that we have observed for the RNA sequences is that 
\begin{equation} \label{1.8} z_{max}(x) \leq z_{max}(x-1) +2 \;\; \text{ for} \; \;  2 \leq z \leq z_{max}.\end{equation}
We know that this is not true for all sequences. For example, for the sequence \begin{equation}\label{seq} \text{ ACCCCGACCCUUUUCCCAGCCCCA }\end{equation} we have $z_{max}(2) \geq 6$ but $z_{max}(1) =3$. However, notice that if the rooted tree corresponding to the structure with signature $(x,y,z,w)$ has depth more than 1, there are two branching points separated by a stem and breaking the base pairs in that stem produces a structure with $x-1$ branching points and $z-2$ branches. Therefore, a counterexample would have to satisfy the property that for some $x > 1$, all structures with $x$ branching points and $z_{max}(x)$  branches do not have a path between the branching points that does not involve the root. The number of such type of structures is limited because of the possibility
of alternative configurations, but a detailed discussion (and certainly proof) involve a different approach, so won't be given here.

Suppose~\eqref{1.8} is satisfied. Then $z_{max} \leq z_{max}(1) + 2(x-1)$ for $1 \geq x \geq x_{max}$. Let $(x,y,z,w) \in \sigs_{1}$ be such that $R = \cone(x, y, z, w) \cap \acpl \neq \emptyset$ and $z=z_{max}(x) = z_{max}(1) + 2(x-1)$. Then for $(x',y',z',w') \in \sigs_{0}$, since $z' \leq z_{max}(1) + 2(x'-1)$, we have $z-z' \geq 2(x-x')$ which implies 
\[ \max\left\{ \frac{x-x'}{z-z'}: (x',y',z',w') \in \sigs_{0}, z' < z\right\} \leq \frac{1}{2} \leq \min\left\{ \frac{x-x''}{z-z''}: (x'',y'',z'',w'') \in \sigs_{0}, z'' > z\right\}\]
and, therefore, by Theorem~\ref{SE}, $R$ is unbounded. This explains the arithmetic progression with step 2 in the $z$ coordinates that we see in Figure~\ref{fig: figexample} when we traverse the unbounded regions clockwise starting from $(0,0,0,w_{0})$.

\section{Conclusion}
\label{conclusion}
%
%

We have shown that for each sequence, under certain conditions which are empirically true for naturally occurring RNA sequences, the structures with minimal and maximal number of branches for a given number of branching points are optimal with high probability, when the parameter $b$ which penalizes for single stranded nucleotides in the multibranch loops is kept constant. This was done via a complete characterization of the unbounded regions in the $\acpl$ section of the normal fan of the branching polytope. While not all maximal branching structures correspond to unbounded regions, we have completely characterized those that do, and shown that this really depends on the combinatorics of the possible pairings for the sequence, not on the energy of the other motifs in those structures.   

Some of our results depend on assumptions which we have observed to be true for branching polytopes of RNA sequences. We believe that these assumptions need not be true for all sequences over the four letter alphabet, but that the counterexamples would be pathological, like the sequence~\eqref{seq} which is very short and palindromic, for instance. In another case, for the precise condition in Theorem~\ref{SE}, we had to introduce the set $\sigs_{1}$ which is determined by the technical condition~\ref{property3}. However, for the branching polytopes we have computed we have observed that \[x_{1} < x_{2} \implies z_{max}(x_{1}) \leq z_{max}(x_{2})\] and that there is a structure with $x$ branching points and $z$ branches for every $x_{min} \leq x \leq x_{max}$, $z_{min}(x) \leq z \leq z_{max}(x)$. These conditions together imply~\eqref{property3} which means that in practice $\sigs_{1}=\sigs_{0}$.  We expect that a counterexample would also be a pathological sequence. Therefore, the following is natural to ask: \emph{Can our assumptions be mathematically justified?  Is there a reasonable probability distribution of sequences under which Observations~\ref{observmax} and~\ref{observmin} hold?}

As a consequence of out descriptions of the vertices that correspond to the unbounded regions in $\acpl$, we can conclude that the secondary structures that are biologically reasonable have signatures that correspond to bounded regions, where the optimization is less stable under the change of branching parameters. To look at improving the average prediction accuracy, therefore, one would need to consider structures that are approximately correct. The accuracy, stability, and robustness are analyzed in~\cite{inprep}.

\bibliographystyle{plain}
\bibliography{polytopes.bib}


\end{document}